\documentclass[submission,copyright,creativecommons,fleqn]{eptcs}
\usepackage{graphicx} 
\usepackage{amssymb}
\usepackage{color}
\usepackage{amsmath}
\usepackage{amsthm}
\usepackage{bm}
\usepackage{bbm}
\usepackage{subfig}
\newtheorem{theorem}{Theorem}[section]
\newtheorem{definition}{Definition}
\providecommand{\event}{LSFA 2025} 

\usepackage{iftex}

\ifpdf
  \usepackage{underscore}         
  \usepackage[T1]{fontenc}        
\else
  \usepackage{breakurl}           
\fi

\def\titlerunning{A Function-Set Framework}
\def\authorrunning{L. Bayzid, A. Madeira, \& M. A. Martins}

\title{A Function-Set Framework:\\General Properties and Applications to Modal Logic}
\author{Luke Bayzid
\qquad\qquad Alexandre Madeira \qquad\qquad Manuel A. Martins
\institute{Mathematics Department of University of Aveiro \\
CIDMA - Research Center in Mathematics and Applications, University of Aveiro, Portugal}
\email{luke.adrian@ua.pt \quad madeira@ua.pt \quad\qquad martins@ua.pt}
}
\begin{document}
\setcounter{page}{1}
\maketitle
\newtheorem{ex}{Example}
\begin{abstract}
    Representations are essential to mathematically model phenomena, but there are many options available. While each of those options provides useful properties with which to solve problems related to the phenomena in study, comparing results between these representations can be non-trivial, as different frameworks are used for different contexts. We present a general structure based on set-theoretic concepts that accommodates many situations related to logical and semantic frameworks. We show the versatility of this approach by presenting alternative constructions of modal logic; in particular, all modal logics can be represented within the framework.
\end{abstract}

\section{Introduction}

When attempting to model a phenomenon, we have access to a variety of different representations, many in completely distinct areas of mathematics, computer science, and beyond \cite{SpecBook,ZBook}. This diversity of models is helpful because the particularity of each representation grants useful properties to study appropriate objects. However, objects represented in distinct areas can be hard to compare or program in a standardized form, especially if we wish to do so constructively. 

The main motivation for this short paper, which is part of the first author's master's thesis \cite{Lukethesis}, is to define a framework based on entities and time, applying it to modal logic thereafter. Our goal, as such, was to generalize what these concepts mean while still capturing broad properties. Since set theory is a well-studied area of mathematics with a wide reach, we chose it as the basis.

Before formally defining the system, we believe that it is important to slowly build up the intuition that led to the creation thereof. For this, we can start with a practical example and generalize it until we arrive at our system.

Let us suppose that we wish to represent a town of people who go to work during certain parts of the day. Already, we have made several assumptions — that there is a notion of time, actions, and people. These three concepts form the foundation of our system. In fact, any question about the described collection can be seen as a question about an entity's state at a point in time; if we wished to include the environment, it, too, could be considered an entity or merely something encoded within the states.

Now that we have extracted the objects of our system, we should formalize them and give them properties: Time is typically treated as the set of natural numbers, but the most important aspect is that it is totally ordered. States encode the information that we wish to study, and as that can vary significantly, we cannot particularize them more than just to say that they form a set. Since the semantics of the system are encoded within the state, the only important part about entities is that they are distinguishable; as such, we only need to store a set of cardinals — an indexing system, much like a computer's memory addresses. It should be noted that even if we do not care about encoding time, we can make it a singleton.

It should now be clear that we are a approaching a temporal, entity-based indexing system that points to the semantics of an object at a particular moment. This is not the only generalization of the above, but it is the one that we followed. With these intuitions set, we will now formally define the above system in Section~ \ref{Framework}, use the definition to describe a common property in Section~\ref{Determinability}, and apply the framework to modal logic in Section~\ref{Modal Context}, where we prove that modal logic can, in a sense, be seen as a particularization thereof in Theorem~\ref{Main}. We provide relevant future work to show the framework's direction in Section~\ref{Conclusion}.

\section{Framework}\label{Framework}

Let $S\not=\emptyset$ be a set of states, $E\not=\emptyset$ be a set of entities, and $T\not=\emptyset$ be a totally ordered set. Then, we call $\Omega\subseteq S^{E\times T}$ a \emph{context}, which is the primary object of study; an element of a context is called an \emph{instance}. Intuitively, each instance represents a possible timeline, where its corresponding context represents the set of all possible timelines. As is commonly known, we also have that $S^{E\times T}\cong {S^E}^T\cong {S^T}^E$, which will be useful later. In this context, given a $\omega\in S^{E\times T}$ (i.e., a function $\omega: E\times T \to S$), we write
\begin{itemize}
    \item $\omega^*$ for the corresponding element of ${S^E}^T$ (i.e., the function $\omega^*:T \to (E \to S)$) and
    \item $\omega^@$ for the corresponding element of ${S^T}^E$ (i.e., the function $\omega^@:E \to (T \to S)$),
\end{itemize} whereas a subscript of the same type is used to take an element of ${S^E}^T$ or ${S^T}^E$ and turn it into one of $S^{E\times T}$, such as $\omega_*$ or $\omega_@$.
\begin{ex}
Let us consider a trivial example to illustrate the notation — Alice and Bob live together, and it is known that when Bob is at home, Alice will be at home an hour later. We can define the aforementioned situation as the following:\begin{center}
$S=\{\mathrm{Home}, \mathrm{Out}\}$;\\
$E=\{\mathrm{Alice}, \mathrm{Bob}\}$;\\
$T=\mathbb{N}$; and\\
$\Omega=\{\omega\in S^{E\times T}\ |\ \forall t\in T, \omega(\mathrm{Bob}, t)=\mathrm{Home}\Rightarrow \omega(\mathrm{Alice}, t+1)=\mathrm{Home}\}$.
\end{center}
If we then wished to add that Bob will be at home during odd hours, we could write the following:\begin{center}
    $\Omega'= \{\omega\in \Omega\ |\ \forall t\in T, t\ \mathrm{mod}\ 2 = 1\Rightarrow \omega(\mathrm{Bob}, t)=\mathrm{Home}\}$.
\end{center}
\end{ex}
\begin{definition}
    When analyzing an instance up to a particular point, it is typically useful to know its neighborhood — similar possibilities permitted by the context. For this, we use the \emph{consistency context} of $\tilde\omega$ up to $t$, which is defined as follows:
\begin{center}
    $\mathrm{C}^{\tilde\omega, t}(\Omega)=\{ \omega\in\Omega\ |\ \forall e\in E, \forall t'\in T, t'\leq t \Rightarrow \omega(e, t')=\tilde\omega(e, t')\}$.
\end{center}
\end{definition}
Intuitively, the above definition captures all of the instances that agree with $\tilde\omega$ up to $t$; a similar function can be made if $t$ is meant to be excluded.
\subsection{Determinability}\label{Determinability}
Determinism is an important concept in philosophy \cite{Baumeister2022}, computer science \cite{CSDeterminism}, and physics \cite{vanStrien2021}, as it refers to the determination of the future by using past information. This can be seen as a subset of determining all specific futures that are consistent with past information, which is usually called indeterminism. \begin{definition}
    We model both concepts with contexts under the term \emph{determinability}, which is done using the following:
\begin{eqnarray}
\forall \omega, \omega'\in\Omega,\forall t,t'\in T,\exists\tau\in\mu(t^+, t'^+), \\
\omega^*(t)=\omega'^*(t')\Rightarrow\beta(\tau)\land\{\tilde\omega^*|_{t^+}\ |\ \tilde\omega\in\mathrm{C}^{\omega, t}(\Omega)\}=\{\tilde\omega^*|_{t'^+}\circ\tau \ |\ \tilde\omega\in\mathrm{C}^{\omega', t'}(\Omega)\},
\end{eqnarray}
where $t^+ = \{t'\in T\ |\ t'\geq t \}$, where $\mu(A, B)$ is the set of all monotonous functions from $A$ to $B$, and where $\beta(\tau)$ qualifies $\tau$ as a bijection. 
\end{definition}
\noindent
Here, the idea is to say that if two instances have the same state at certain moments, then the futures generated by them are the same (i.e., consistent). In particular, if there is only one future at each point (possibly excluding the minimum of $T$ if it exists), we say that the context is deterministic.
\begin{figure}[h]
        \centering
        \subfloat[Consistency Context]{\includegraphics[width=0.25\linewidth]{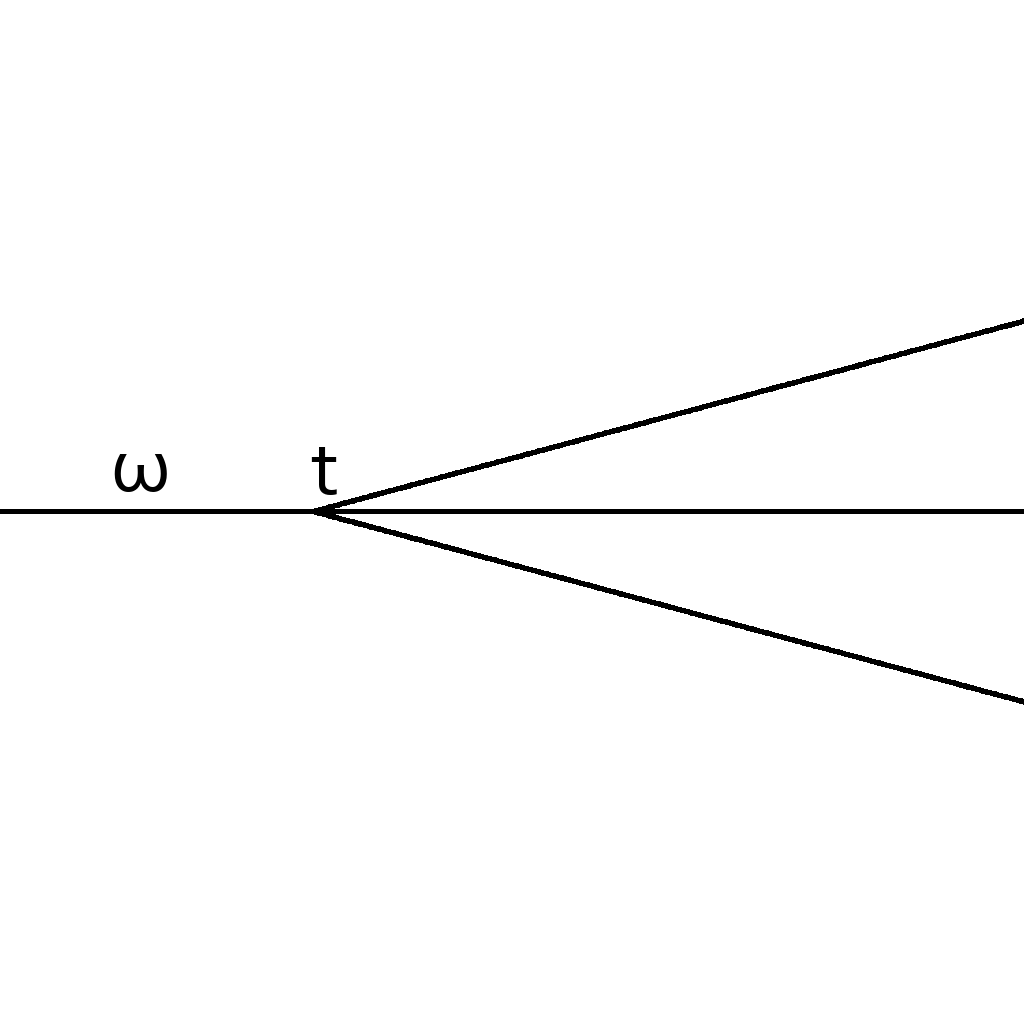}}\hspace{2cm}
        \subfloat[Determinability]{\includegraphics[width=0.25\linewidth]{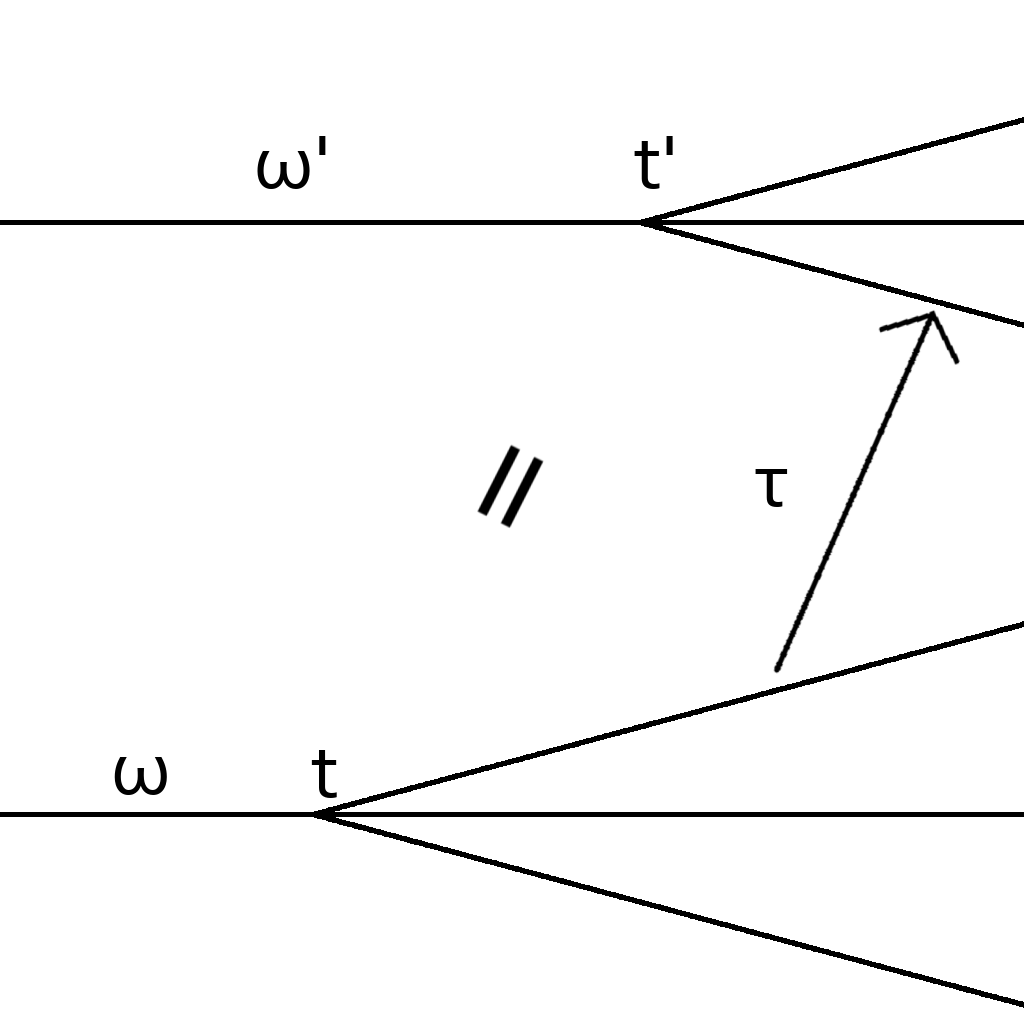}}
\end{figure}
\begin{ex}
Let us consider $\mathrm{Chess}\subseteq (({\mathbf{8}^2}\cup\mathbf{1})^{16})^{\mathbf{2}\times\mathbb{N}}$. The idea is to encode chess as a function that keeps track of the pieces ($16$) of each player ($\mathbf{2}$), marking them on the board ($\mathbf{8}^2$) or outside of it ($\mathbf{1}$) over $\mathbb{N}$. With this encoding, chess is not determinable. One reason is that we cannot know who is next by looking at the state of the board. To amend that, we could include another entity that keeps track of the turns; alternatively, we could extend the states, such as $(({\mathbf{8}^2}\cup\mathbf{1})^{16}\times\mathbf{2})^{\mathbf{2}\times\mathbb{N}}$, such that a player can only make their turn if the second coordinate (the added portion) is one, which would alternate as turns progress. Beyond this, certain rules of chess might require similar changes to grant determinability.
\end{ex}
\paragraph{Determinability With Countable Time.}
In contexts with countable time, we can actually use a slightly different but intuitive understanding of determinability. In a sense, we expect any countable, determinable system to be describable by a function that iterates the present into the future. 
\begin{definition}
    Let $A=\{\omega^*(t)\ |\ (t, \omega)\in T\times\Omega\}$. We say that $\Omega$ has an iterator whenever $\exists i\in \mathrm{P}(A)^A,\forall\omega\in\Omega,\forall t\in T, (C^{\omega, t}(\Omega))^*(t+1)=i(\omega^*(t))$.
\end{definition}
\begin{theorem}
     Let $T$ be countable. Let $\Omega\subseteq S^{E\times T}$ be determinable. Then, $\Omega$ has an iterator.
\end{theorem}
\begin{proof}
    
Let us say that $i$ is a binary relation of the form $A\times P(A)$ such that
\[(a, b)\in i\Rightarrow a\in\Omega^*(t)\land b=(C^{\omega, t}(\Omega))^*(t+1)\]
where $\omega\in\Omega$ such that $\omega^*(t)=a$ for some $t\in T$. We will show that $i$ really is a function, for which we must prove that if $a=a'$ in $(a, b)\in i\land(a', b')\in i$, then $b=b'$. Since $a=a'$, we have that $\exists\omega,\omega'\in\Omega, \exists t,t'\in T,\omega^*(t)=a=a'=\omega'^*(t')$. Given determinability, we know that $\exists\tau\in\mu(t^+,t'^+),$
\begin{center}
    $\beta(\tau)\land\{\tilde\omega^*|_{t^+}\ |\ \tilde\omega\in\mathrm{C}^{\omega, t}(\Omega)\}=\{\tilde\omega^*|_{t'^+}\circ\tau \ |\ \tilde\omega\in\mathrm{C}^{\omega', t'}(\Omega)\}$.
\end{center}
\noindent
By bijective monotony over a total order, since $t$ is the minimum of $t^+$, it must also get mapped by $\tau$ to the minimum of $t'^+$, which is $t'$, for which the same applies to $t+1$ and $t'+1$ by virtue of bijectivity and countability. As such, we can particularize the above to 
\begin{center}
    $\{\tilde\omega^*(t+1)\ |\ \tilde\omega\in\mathrm{C}^{\omega, t}(\Omega)\}=\{\tilde\omega^*(t'+1) \ |\ \tilde\omega\in\mathrm{C}^{\omega', t'}(\Omega)\}$,
\end{center}
which is to say that $(C^{\omega, t}(\Omega))^*(t+1)=(C^{\omega', t'}(\Omega))^*(t'+1)$. Therefore, $b=b'$, meaning that $i$ truly is a function.
\end{proof}
What is slightly harder to prove is that the converse is true, meaning that having an iterator is equivalent to being determinable in contexts with countable time, which confirms our intuitions.
\begin{theorem}
     Let $T$ be countable. Let $\Omega\subseteq S^{E\times T}$ have an iterator. Then, $\Omega$ is determinable.
\end{theorem}
\begin{proof}
Let us assume that $\Omega$ has $i$. Let us also assume that $\omega^*(t)=\omega'^*(t')$. Since $i$ exists, we have that
\begin{center}
    $i(\omega^*(t))= (C^{\omega, t}(\Omega))^*(t+1)=(C^{\omega', t'}(\Omega))^*(t'+1)=i(\omega'^*(t'))$.
\end{center}
This means that
\begin{center}
    $\{\tilde\omega^*(t+1)\ |\ \tilde\omega\in\mathrm{C}^{\omega, t}(\Omega)\}=\{\tilde\omega^*(t'+1) \ |\ \tilde\omega\in\mathrm{C}^{\omega', t'}(\Omega)\}$.
\end{center}
This both proves the result for the base case of $t^+$ and $t'^+$ and can be inductively applied to the remaining elements thereof (i.e., use $i(\tilde\omega^*(t))$ to prove the result for $t+1$, inductively proving the result for $t^+$ for all $t\in T$). With $\cdot+1$, we create a monotonous bijection from $t^+$ to $t'^+$: $\forall n\in\mathbb{N}_1, \tau:t+n\mapsto t'+n$.
\end{proof}
\section{Modal Context}\label{Modal Context}

In this section, we will focus on modeling standard modal logic \cite{Blackburn2001} using contexts and show that this framework is a generalization thereof, allowing us to capture the internal aspects of a Kripke frame. We call our constructions \emph{modal contexts} because we treat each instance as a modal world, where standard modal concepts are used. Although this section focuses on standard modal logic, the use of contexts \cite{Lukethesis} also reaches dynamic logic \cite{DL} and epistemic logic \cite{EL}.

\paragraph{Modal Operators.}
From an abstract perspective, the point of using a context is to treat the states as black boxes. For this, we have to generalize the notion of applying a modal operator. We do this by creating a \emph{modal operator}, which is an injection $\square: S\rightarrow S$ such that\begin{center}
    $\forall s\in S,\forall n\in\mathbb{N}, \square^{n+2}(s)\not=\square(s)$,
\end{center} which is to say that there are no loops (e.g., $\square(\square s)\not=\square s$). Assuming $S\not=\emptyset$, this naturally requires $|S|\geq\aleph_0$. This is done so that we can encode that there is a syntactic distinction between $\square\phi$ and $\phi$.

\paragraph{Modal Context.}
\begin{definition}
    Let us consider a power context $M\subseteq\mathbb{W}=\mathrm{P}(S)^{E\times T}$. Let $\square$ and $\lozenge$ be modal operators. Let $R$ be a binary relation on $\mathbb{W}$. We say that $M$ is a modal context whenever $\forall\omega\in\mathbb{W},\omega\in M\iff$
    \begin{center}
    $((\forall e\in E,\forall t\in T,\forall s\in S,(\square s\in\omega(e, t)\iff(\forall\omega'\in\omega R\cdot,s\in\omega'(e, t))))$\\$\land$\\$(\forall e\in E,\forall t\in T,\forall s\in S,(\lozenge s\in\omega(e, t)\iff(\exists\omega'\in\omega R\cdot,s\in\omega'(e, t)))))$.
\end{center}
\end{definition}
The primary idea here is that we apply modal operators to each of the states whenever related worlds share that underlying state. As we would hope, traditional modal logic can be modeled within this. To prove this, we relate each world (modal logic) to an instance. In particular, we say that $w\in W$ relates to $w'\in M$ whenever\begin{center}
    $\forall\phi\in\Phi, (M,w\models\phi)\iff\phi\in w'(0,0)$.
\end{center}This naturally defines how we can prove statements in a modal context — To prove $\phi$, we merely need to show that it belongs to a specific set.
\begin{theorem}\label{Main}
Let $M=\langle W, R, V\rangle$ be a modal logic. Then, there is a modal context $M'\subseteq\mathbb{W}=\mathrm{P}(\Phi)^{\mathbbm{1}\times\mathbbm{1}}$ such that\begin{center}
    $\forall w\in W,\exists w'\in M',\forall\phi\in\Phi, (M,w\models\phi)\iff\phi\in w'(0,0)$.
\end{center}
\end{theorem}
\begin{proof}
For this proof, we will focus on constructing $M'$ and show that it is a modal context.
We will first create an equivalence class over $W$. We do this by stating that \begin{center}
    $wQw'\iff(\forall\phi\in\Phi, (M,w\models\phi)\iff (M,w'\models\phi))$,
\end{center} where reflexivity, symmetry, and transitivity are all trivial to prove.
We can now define $M'$ by stating that $\forall\omega\in\mathbb{W}$,\begin{center}
    $\omega\in M'\iff(\exists[\omega']\in W/_{Q}, \omega(0,0)=\{\phi\in\Phi\ |\ M,\omega'\models\phi\})$.
\end{center} Since each function is fully defined by the set to which it maps $(0, 0)$ (all of whom are equal by class), the functions are in bijection with $W/_{Q}$. As such, we will denote the class of $W/_{Q}$ that is associated with $\omega$ by $[\omega]$.

To finish our preliminary constructions, we need to make create $R'$, which we define as\begin{center}
    $\forall\omega,\omega'\in M',\omega R'\omega'\iff (\exists a\in[\omega],\exists b\in[\omega'],aRb)$.
\end{center}

We now just need to prove that $M'$ is a modal context. We will start with the proof of $\square$, which is to say that\begin{center}
    $\forall\omega\in M',\square\phi\in\omega(0, 0)\iff(\forall\omega'\in\omega R'\cdot,\phi\in\omega'(0, 0))$.
\end{center}

We will now prove $\Rightarrow$, meaning that we will assume that $\square\phi\in\omega(0,0)$: Let $\omega\in M'$. If $\square\phi\in\omega(0,0)$, we have that\begin{center}
    $\forall\omega'\in[\omega], (M,\omega'\models\square\phi)$.
\end{center} Then, by the properties of modal logic, we have that\begin{center}
    $\forall\omega'\in[\omega],\forall\omega''\in\omega'R\cdot,(M,\omega''\models\phi)$.
\end{center} Let $\omega''\in\omega R'\cdot$, which means that $\exists\omega'\in[\omega],\exists\omega'''\in[\omega''],\omega' R\omega'''$. However, with the above, that means that $M,\omega'''\models\phi$, which means that $\phi\in\omega''(0,0)$, concluding our proof of $\Rightarrow$.

We will now prove $\Leftarrow$, meaning that we will assume that $\forall\omega'\in\omega R'\cdot,\phi\in\omega'(0, 0)$: Let $\omega'\in[\omega]$. Let $\omega''\in\omega'R\cdot$. As such, we know that $\exists\omega'''\in W/_{Q}, \omega''\in[\omega''']$ such that $\omega R'\omega'''$. Therefore, using our assumption, $\phi\in\omega'''(0,0)$, meaning that $\forall\omega''''\in[\omega'''],(M,\omega''''\models\phi)$ (in particular, $\omega''$). Since this was done generically for $\omega''\in\omega'R\cdot$, we have that $\forall\omega''\in\omega'R\cdot,(M,\omega''\models\phi)$, which means that $M,\omega'\models\square\phi$ by the properties of modal logic. As such, $\square\phi\in\omega(0,0)$, concluding our proof of $\Leftarrow$.

We now need to prove $\lozenge$, which is to say that\begin{center}
    $\forall\omega\in M',\lozenge\phi\in\omega(0, 0)\iff(\exists\omega'\in\omega R'\cdot,\phi\in\omega'(0, 0))$.
\end{center} Similarly to what was done above, we will divide this proof into two parts: $\Rightarrow$ and $\Leftarrow$.

We will now prove $\Rightarrow$, meaning that we will assume that $\lozenge\phi\in\omega(0,0)$: Let $\omega\in M'$. If $\lozenge\phi\in\omega(0,0)$, we have that\begin{center}
    $\forall\omega'\in[\omega], (M,\omega'\models\lozenge\phi)$.
\end{center} Then, by the properties of modal logic, we have that\begin{center}
    $\forall\omega'\in[\omega],\exists\omega''\in\omega'R\cdot,(M,\omega''\models\phi)$.
\end{center} Let us choose $\omega'\in[\omega]$. We now have the above $\omega''$, whose class in $W/_{Q}$ we will call $[\omega''']$; in particular, we can say that $\omega'''\in M'$ because $W/_{Q}$ is in bijection with $M'$. $\omega'''\in\omega R'\cdot$ because $\omega'R\omega''$. Since $M,\omega''\models\phi$, we have that $\phi\in\omega'''(0,0)$, which proves that $\exists\omega'''\in\omega R'\cdot,\phi\in\omega'''(0, 0)$, concluding our proof of $\Rightarrow$.

We will now prove $\Leftarrow$, meaning that we will assume that $\exists\omega'\in\omega R'\cdot,\phi\in\omega'(0, 0)$: If $\omega'\in\omega R'\cdot$, then $\exists a\in[\omega],\exists b\in[\omega'],aRb$, and since $b\in[\omega']$, we have that $M,b\models\phi$, which means that $M,a\models\lozenge\phi$. As such, $\lozenge\phi\in\omega(0,0)$, concluding our proof of $\Leftarrow$.

At last, we have proven that $M'$ is a modal context whose instances capture the worlds of $M$.
\end{proof}

An interesting note about the above theorem is that we can faithfully capture modal logic without using labels for worlds. Moreover, this approach generically works for any power context.

\section{Conclusion and Future Work}\label{Conclusion}

Contexts allow us to create by restriction — We choose a function set that is large enough and choose a subset to fit the situation. This perspective is in contrast to many standard approaches, such as a Kripke frame, where we might have to define the underlying structure and carefully add rules or components. This approach might not be adequate for all situations, but it may allow for other views on existing phenomena.

We have shown that the framework can be used to abstractly model structures without worrying about many details (such as what was done with chess and determinability). We have also shown that we can encode standard modal logic into a particularization of the framework, with Theorem \ref{Main} to constructively convert the standard axioms of modal logic.

However, this implementation of modal logic does not use $E$ and $T$, which are primary points of interest for contexts; in fact, this practically just makes the context a set. This is because it is generally doubtful what $E$ and $T$ would mean unless we are discussing epistemic or temporal logic. We would like to investigate extensions of this definition that are also consistent with $E$ and $T$, which, in the case of epistemic logic, is to say that the instance mapping of $e\in E$ should correspond to, for example, its knowledge as a power set, such as \begin{center}
    $\omega(\mathrm{Alice},t)=\{\text{"Bob bought a house."},\text{"Bob knows that Alice knows that Bob bought a house."},\dots\}$.
\end{center} In \cite{Lukethesis}, we have an alternative definition to solve this for epistemic logic, though it is rather independent of how modal contexts were defined here. Temporal logic will also be approached, as the properties of $T$ being a total order are typically sufficient for many common temporal systems.

Much of what has been presented here forms the basics of this framework, and there are many unanswered questions that we would like to explore in future work. This includes the following:
\begin{itemize}
    \item Study function sets with broader kinds of domains, not just $E$ and $T$. In particular, studying how the properties of the domain (such as $T$ being a total order) affect the resulting context is the main interest.    
    \item Model epistemic \cite{EL}/doxastic \cite{doxastic}/temporal \cite{TL} /dynamic \cite{DL} logic in a way that is coherent with $E$ and $T$, meaning that each $e\in E$ and $t\in T$ should correspond to their modal counterparts. This may possibly require other constructions of modal contexts.
    \item Define bisimulations and modal invariance for modal contexts.
    \item Find the conditions under which expected properties (e.g., finite model property and decidability) still hold and techniques (e.g., filtration) can be applied.
\end{itemize}

\paragraph{Acknowledgments.}

This work is supported by FCT – Fundação para a Ciência e a Tecnologia through projects UIDB/04106/2025 at CIDMA and by National and European Funds through SACCCT- IC\&DT - Sistema de Apoio à Criação de Conhecimento Científico e Tecnológico, as part of COMPETE2030, within the project BANSKY with reference number 15253.


\end{document}